\documentclass[10pt,twocolumn,letterpaper]{article}
\usepackage{cvpr}
\usepackage{times}
\usepackage{epsfig}
\usepackage{graphicx}
\usepackage{amsmath}
\usepackage{amssymb}

\usepackage{amsthm}
\usepackage{thmtools}
\usepackage{etex}
\usepackage{amsxtra}
\usepackage{amsfonts}
\usepackage{nicefrac}
\usepackage{microtype}
\usepackage{url}
\usepackage{soul}

\usepackage{color}
\usepackage{bbm}
\usepackage{pgfplots}
\usepackage{tikz}
\usetikzlibrary{datavisualization}

\usepackage{colortbl}
\newcolumntype{g}{>{\columncolor[gray]{0.85}}c}

\usepackage{enumitem}

\usepackage{wrapfig}
\usepackage{sidecap}
\usepackage{floatrow}
\usepackage{standalone}
\usepackage[ruled,linesnumbered]{algorithm2e}

\usepackage{stmaryrd} 
\usepackage{mathtools}

\usepackage{multirow}
\usepackage{fixmath}
\usepackage{booktabs}
\usepackage{pbox}
\usepackage{subfigure}

\usepackage[pagebackref=true,breaklinks=true,letterpaper=true,colorlinks,bookmarks=false]{hyperref}

\newcommand{\SN}{\mathcal{N}}
\newcommand{\SG}{\mathcal{G}}
\newcommand{\SL}{\mathcal{L}}
\newcommand{\SX}{X}

\newcommand{\BF}{\mathbb{F}}
\newcommand{\BR}{\mathbb{R}}

\newcommand{\BE}{\mathbb{E}}
\newcommand{\BG}{\mathbb{G}}
\newcommand{\BS}{\mathbb{S}}
\newcommand{\BX}{X}

\newcommand{\FG}{\mathsf{G}}
\newcommand{\FV}{\mathsf{V}}
\newcommand{\FE}{\mathsf{E}}

\newcommand{\FC}{\mathsf{C}}
\newcommand{\FL}{\mathsf{L}}
\newcommand{\FX}{X}

\newcommand{\IPSLP}{IRPS-LP}

\newcommand{\R}{\mathbb{R}}
\newcommand{\N}{\mathbb{N}}

\newcommand{\la}{\langle}
\newcommand{\ra}{\rangle}
\newcommand{\T}{\top}
\newcommand{\eins}{\mathbbmss{1}}
\newcommand{\dcup}{\dot{\cup}}

\DeclareMathOperator*{\argmin}{arg min}
\DeclareMathOperator*{\argmax}{arg max}
\DeclareMathOperator*{\conv}{conv}

\providecommand{\abs}[1]{\lvert#1\rvert}

\definecolor{fettrot}{RGB}{255,10,10}

\newtheorem{definition}{Definition}
\newtheorem{theorem}{Theorem}
\newtheorem{lemma}{Lemma}
\newtheorem*{lemma*}{Lemma}
\newtheorem{proposition}{Proposition}
\newtheorem*{proposition*}{Proposition}
\newtheorem*{theorem*}{Theorem}
\newtheorem{remark}{Remark}

\declaretheoremstyle[bodyfont=\normalfont]{normalbody}
\declaretheorem[style=normalbody,name=Example]{example}
\newcommand{\myparagraph}[1]{\paragraph{#1}}

\makeatletter
\def\smallunderbrace#1{\mathop{\vtop{\m@th\ialign{##\crcr
   $\hfil\displaystyle{#1}\hfil$\crcr
   \noalign{\kern3\p@\nointerlineskip}%
   \tiny\upbracefill\crcr\noalign{\kern3\p@}}}}\limits}
\makeatother

\cvprfinalcopy 

\begin{document}

\title{A Dual Ascent Framework for Lagrangean Decomposition of Combinatorial Problems}


\author{Paul Swoboda, Jan Kuske, Bogdan Savchynskyy}

\maketitle

\begin{abstract}
 We propose a general dual ascent framework for Lagrangean decomposition of combinatorial problems.
 Although methods of this type have shown their efficiency for a number of problems, so far there was no general algorithm applicable to multiple problem types. 
 In this work, we propose such a general algorithm. It depends on several parameters, which can be used to optimize its performance in each particular setting.
 We demonstrate efficacy of our method on graph matching and multicut problems, where it outperforms state-of-the-art solvers including those based on subgradient optimization and off-the-shelf linear programming solvers.
\end{abstract}

\section{Introduction}

Computer vision and machine learning give rise to a number of powerful computational models.  
It is typical that inference in these models reduces to non-trivial combinatorial optimization problems.
For some of the models, such as conditional random fields (CRF), powerful specialized solvers like~\cite{TRWSKolmogorov,SRMPKolmogorov,boykov2001fast,lempitsky2010fusion} were developed.  
In general, however, one has to resort to off-the-shelf integer linear program (ILP) solvers like CPLEX~\cite{cplex} or Gurobi~\cite{gurobi}. 
Although these solvers have made a tremendous progress in the past decade, the size of the problems they can tackle still remains a limiting factor for many potential applications, as the running time scales super-linearly in the problem size. The goal of this work is to partially fill this gap between practical requirements and existing computational methods.

It is an old observation that many important optimization ILPs can be efficiently decomposed into easily solvable combinatorial sub-problems~\cite{guignard1987lagrangean-TandA}. The convex relaxation, which consists of these sub-problems coupled by linear constraints is known as Lagrangean or dual decomposition~\cite{guignard1987lagrangean,DualDecompositionKomodakis}. Although this technique can be efficiently used in various scenarios to find approximate solutions of combinatorial problems, it has a major drawback: In the most general setting only slow (sub)gradient-based techniques~\cite{lemarechal1992lagrangian,necoara2008application,DualDecompositionKomodakis,kappes2012bundle,savchynskyy2011study} can be used for optimization of the corresponding convex relaxation. 

In the area of conditional random fields, however, it is well-known~\cite{OpenGMBenchmark} that message passing or dual (block-coordinate) ascent algorithms (like e.g.~TRW-S~\cite{TRWSKolmogorov}) significantly outperform (sub)gradient-based methods. 
Similar observations were made much earlier in~\cite{ribeiro1986solving} for a constrained shortest path problem.

Although dual ascent algorithms were proposed for a number of combinatorial problems (see the related work overview below), there is no general framework, which would 
(i)~give a generalized view on the properties of such algorithms and more importantly (ii)~provide tools to easily construct such algorithms for new problems. Our work provides such a framework.

\myparagraph{Related Work}
Dual ascent algorithms optimize a dual problem and guarantee monotonous improvement (non-deterioration) of the dual objective. The most famous examples in computer vision are block-coordinate ascent (known also as {\em message passing}) algorithms like TRW-S~\cite{TRWSKolmogorov} or MPLP~\cite{MPLP} for maximum a posteriori inference in conditional random fields~\cite{OpenGMBenchmark}.

To the best of our knowledge the first dual ascent algorithm addressing integer linear programs belongs to Bilde and Krarup~\cite{bilde1977sharp} (the corresponding technical report in Danish appeared $1967$). In that work an uncapacitated facility location problem was addressed. A similar problem (simple plant location) was addressed with an algorithm of the same class in~\cite{guignard1988lagrangean}. 
In $1980$ Fisher and Hochbaum~\cite{fisher1980database} constructed a dual ascent-based algorithm for a problem of database location in computer networks, which was used to optimize the topology of Arpanet~\cite{ArpaNet}, predecessor of Internet.
The generalized linear assignment problem was addressed by the same type of algorithms in~\cite{fisher1986multiplier}. The Authors considered a Lagrangean decomposition of this problem into multiple knapsack problems, which were solved in each iteration of the method. An improved version of this algorithm was proposed in~\cite{guignard1989technical}. Efficient dual ascent based solvers were also proposed for the min-cost flow in~\cite{BeliefPropagationNetworkFlow}, for the set covering and the set partitioning problems in~\cite{fisher1990optimal} and the resource-constrained minimum weighted arborescence problem in~\cite{guignard1990application}.
The work~\cite{guignard1989application} describes basic principles for constructing dual ascent algorithms. Although the authors provide several examples, they do not go beyond that and stick to the claim that these methods are structure dependent and problem specific.

The work~\cite{CombinatorialOptimizationMaxProductDuchi} suggests to use the max-product belief propagation~\cite{weiss2001optimality} to decomposable optimization problems. However, their algorithm is neither monotone nor even convergent in general.

In computer vision, dual block coordinate ascent algorithms for Lagrangean decomposition of combinatorial problems were proposed for multiple targets tracking~\cite{arora2013higher}, graph matching (quadratic assignment) problem~\cite{HungarianBP} and inference in conditional random fields~\cite{TRWSKolmogorov,SRMPKolmogorov,MPLP,Werner07,Werner10,AdaptiveDiminishingSmoothingSavchynskyy,ConvergentDecompositionSolversJancsary,ConvergentMessagePassingAUnifiedView,SubproblemTreeCalibrationWang}. From the latter, the TRW-S algorithm~\cite{TRWSKolmogorov} is among the most efficient ones for pairwise conditional random fields according to~\cite{OpenGMBenchmark}. The SRMP algorithm~\cite{SRMPKolmogorov} generalizes TRW-S to conditional random fields of arbitrary order. In a certain sense, our framework can be seen as a generalization of SRMP to a broad class of combinatorial problems.

\myparagraph{Contribution.} 
We propose a new dual ascent based computational framework for combinatorial optimization.
To this end we:\\
(i) Define the class of problems, called {\em integer-relaxed pairwise-separable linear programs} (\IPSLP), our framework can be used for. Our definition captures Lagrangean decompositions of many known discrete optimization problems (Section~\ref{sec:FactorLP}). \\
(ii) Give a general monotonically convergent message-passing algorithm for solving \IPSLP, which in particular subsumes several known solvers for conditional random fields
(Section~\ref{sec:MessagePassingAlgorithm}).\\
(iii) Give a characterization of the fixed points of our algorithm, which subsumes such well-known fixed point characterizations as {\em weak tree agreement}~\cite{TRWSKolmogorov} and {\em arc-consistency}~\cite{Werner07} (Section~\ref{sec:FixedPoints}).

We demonstrate efficiency of our method by outperforming state-of-the-art solvers for two famous special cases of \IPSLP, which are widely used in computer vision: the multicut and the graph matching problems.
(Section~\ref{sec:experiments}).

A C++-framework containing the above mentioned solvers and the datasets used in experiments are available under \url{http://github.com/pawelswoboda/LP_MP}.

We give all proofs in the supplementary material.

\myparagraph{Notation.}
Undirected graphs will be denoted by $G=(V,E)$, where $V$ is a finite {\em node set} and $E\subseteq{{V}\choose{2}}$ is {\em the edge set}.
The set of neighboring nodes of $v \in V$ w.r.t. graph $G$ is denoted by $\SN_G(v) := \{ u: uv \in E\}$.
The convex hull of a set $\SX \subset \R^n$ is denoted by $\conv(\SX)$. Disjoint union is denoted by $\dot{\cup}$.

\section{Integer-Relaxed Pairwise-Separable Linear Programs (\IPSLP)}
\label{sec:FactorLP}
%
%
Combinatorial problems having an objective to minimize some cost $\theta(x)$ over a set $X\subseteq\{0,1\}^n$ of binary vectors often have a decomposable representation as 
$\min_{x_i\in X_i\atop i=1,\dots,k} \sum_{i=1}^k \la \theta_i,x_i \ra$ for $X_i\subseteq\{0,1\}^{d_i}$ being sets of binary vectors, typically corresponding to subsets of the coordinates of $X$. This decomposed problem is equivalent to the original one under a set of linear constraints $A_{(i,j)}x_i = A_{(j,i)} x_j$, which guarantee the mutual consistency of the considered components. Replacing $X_i$ by its convex hull $\conv(X_i)$ and therefore switching to real-valued vectors from binary ones one obtains a {\em convex relaxation}\footnote{More precisely, this is a linear programming relaxation, since a convex hull of a finite set can be represented in terms of linear inequalities} of the problem, which reads:
\begin{equation}
\label{eq:FactorGraphPrimal}
\min_{\mu \in \Lambda_{\BG}} \sum_{i=1}^k \la \theta_i,\mu_i \ra\,,\ \text{where}\ \Lambda_{\BG}\ \text{is defined as}
\end{equation}
{\small
\begin{equation}\label{eq:Lambda_G}
\Lambda_{\BG} :=
\left\{
(\mu_1\dots\mu_k)\left|
\begin{array}{ll}
\mu_i \in \conv(\BX_i) & i\in \BF \\
A_{(i,j)}\mu_i = A_{(j,i)} \mu_j & \forall ij \in \BE
\end{array}
\right\}\right..
\end{equation}
}
Here $\BF:= \{1,\ldots,k\}$ are called \emph{factors} of the decomposition
and $\BE \subseteq \begin{pmatrix} \BF \\ 2 \end{pmatrix}$ are called \emph{coupling constraints}.
The undirected graph $\BG = (\BF,\BE)$ is called \emph{factor graph}.
We will use variable names $\mu$ whenever we want to emphasize $\mu_i \in \conv(\BX_i)$ and $x$ whenever $x_i \in \BX_i$, $i\in\BF$.


\begin{definition}[IRPS-LP]
   \label{def:IPSLP}
   Assume that for each edge $ij \in \BE$ the matrices of the coupling constraints $A_{(i,j)}, A_{(j,i)}$ are such that
$A_{(i,j)} \in \{0,1\}^{K \times d_i}$ and $A_{(i,j)} x_i \in \{0,1\}^K$  $\forall x_i \in \BX_i$ for some $K \in \N$, analogously for $A_{(j,i)}$.
The problem $\min_{\mu \in \Lambda_{\BG}} \sum_{i\in\BF} \la \theta_i, \mu_i \ra$ is called an 
\emph{\textbf{I}nteger-\textbf{R}elaxed \textbf{P}airwise-\textbf{S}eparable \textbf{L}inear \textbf{P}rogram}, abbreviated by \emph{\IPSLP}.
\end{definition}

In the following, we give several examples of \IPSLP.
To distinguish between notation for the factor graph of \IPSLP, where we stick to bold letters (such as $\BG$, $\BF$, $\BE$) we will use the straight font (such as $\FG$, $\FV$, $\FE$) for the graphs occurring in the examples.

\begin{example}[MAP-inference for CRF]
   \label{example:LocalPolytope}
   A conditional random field is given by a graph $\FG = (\FV,\FE)$, a discrete label space $\FX = \prod_{u \in\FV} \FX_u$, unary $\theta_u : \FX_u \rightarrow \R$ and pairwise costs $\theta_{uv} : \FX_u \times \FX_v \rightarrow \R$ for $u \in \FV$, $uv \in \FE$. 
   We also denote $X_{uv}:=X_u\times X_v$.
The associated {\em maximum a posteriori (MAP)-inference problem} reads
\begin{equation}
\min_{x \in \FX} \sum\nolimits_{u \in \FV} \theta_u(x_u) + \sum\nolimits_{uv \in \FE} \theta_{uv}(x_{uv})\,,
\end{equation}
where $x_u$ and $x_{uv}$ denote the components corresponding to node $u \in \FV$ and edge $uv\in\FE$ respectively.
   The well-known local polytope relaxation~\cite{Werner07} can be seen as an \IPSLP\ by setting $\BF = \FV \cup \FE$, that is associating to each node $v \in \FV$ \emph{and} each edge $uv \in \FE$ a factor, and introducing two coupling constraints for each edge of the graphical model, i.e. $\BE = \{\{u,uv\}, \{v,uv\} : uv \in \FE\}$. 
    For the sake of notation we will assume that each label $s\in\ X_u$ is associated a unit vector $(0,\dots,0,\smallunderbrace{1}_{\clap{s}},0\dots,0)$ with dimensionality equal to the total number of labels $|X_u|$ and $1$ on the $s$-th position. Therefore, the notation $\conv(X_u)$ makes sense as a convex hull of all such vectors. 
        After denoting an $N$-dimensional {\em simplex} as $\Delta_N:=\{\mu\in\BR_+^N\colon \sum_{i=1}^N \mu_i=1 \}$
the resulting relaxation reads
\begin{equation}\label{eq:ILP:objective}
 \min_{\mu\in\FL_{\FG}}  \quad\langle \theta,\mu\rangle:=\sum_{u\in\FV}\la \theta_u,\mu_u\ra + \sum_{uv\in\FE}\la \theta_{uv},\mu_{uv}\ra 
\end{equation}
in the overcomplete representation~\cite{WainwrightBook} and $\FL_{\FG}$ is defined as \\
{\small
\begin{equation}\label{eq:ILP:linear-constraints}
\hspace{-10pt} \begin{array}{ll}
  \underline{\mu_u\in\conv(\SX_{u}):} & \mu_u\in\Delta_{|X_u|},  u\in\FV  \\
  \underline{\mu_{uv}\in\conv(\SX_{uv}):} & \mu_{uv}\in\Delta_{|X_{uv}|}, uv\in\FE\\
  \underline{A_{(uv,u)}\mu_{uv}=A_{(u,uv)}\mu_u:} & \hspace{-7pt}\sum\limits_{x_v\in X_v}\hspace{-3pt}\mu_{uv}(x_u,x_v) = \mu_{u}(x_u),\\ 
   & uv\in\FE, (x_u,x_v)\in X_{uv},\\ & u\in uv, x_u\in X_u\,.
\end{array} 
\end{equation}
}
Here $\mu_u(x_u)$ and $\mu_{uv}(x_u,x_v)$ denote those coordinates of vectors $\mu_u$ and $\mu_{uv}$, which correspond to the label $x_u$ and the pair of labels $(x_u,x_v)$ respectively.
\end{example}   
 
\begin{example}[Graph Matching]\label{example:graph-matching}
        The graph matching problem, also known as {\em quadratic assignment}~\cite{TheQuadraticAssignmentProblem} or {\em feature matching}, can be seen as a MAP-inference problem for CRFs (as in Example~\ref{example:LocalPolytope}) equipped with additional constraints:
        The label set of $\FG$ belongs to a \emph{universe}~$\SL$, i.e. $X_u \subseteq \SL$ $\forall u \in \FV$ and each label $s \in \SL$ can be assigned \emph{at most once}. The overall problem reads
\begin{equation}
   \min_{x} \sum_{u \in \FV} \theta_u(x_u) + \sum_{uv \in \FE} \theta_{uv}(x_u,x_v) \ \text{s.t.} \ x_u \neq x_v \forall u\neq v\,.
\end{equation}
        Graph matching is a key step in many computer vision applications, among them tracking and image registration, whose aim is to find a one-to-one correspondence between image points.
        For this reason, a large number of solvers have been proposed in the computer vision community~\cite{CombinatorialOptimizationMaxProductDuchi,CoveringTreesLowerBoundQuadraticAssignmentJarkony,HungarianBP,GraphMatchingDDTorresaniEtAl,SpectralTechniqueAssignmentLeordeanu,MRFSemidefiniteTorr,ProbabilisticSubgraphMatchingSchellewald,GraduatedAssignmentGold,FactorizedGraphMatching,LocalSparseMatching,IntegerFixedPointGraphMatching,RandomWalksForGraphMatching}.
        Among them two recent methods~\cite{GraphMatchingDDTorresaniEtAl,HungarianBP} based on Lagrangean decomposition show superior performance and provide lower bounds for their solutions. The decomposition we describe below, however, differs from those proposed in~\cite{GraphMatchingDDTorresaniEtAl,HungarianBP}.

Our \IPSLP\ representation for graph matching consists of two blocks: (i) the CRF itself (which further decomposes into node- and edge-subproblems with variables $(\mu_u)_{u \in \FV}$ and (ii) additional {\em label-factors} keeping track of nodes assigned the label $s$. We introduce these label-factors for each label ${s \in \SL}$. The set of possible configurations of this factor $X_s := \{u \in \FV : s \in \SX_u\} \cup \{\#\}$ consists of those nodes $u \in \FV$ which can be assigned the label $s$ and an additional dummy node~$\#$.
The dummy node $\#$ denotes non-assignment of the label $s$ and is necessary, as not every label needs to be taken.
As in Example~\ref{example:LocalPolytope}, we associate a unit binary vector with each element of the set $X_s$, and $\conv(X_s)$ denotes the convex hull of such vectors.
The set of factors becomes $\BF = \FV \dcup \FE \dcup \SL$,
with the set 
$\BE = \{\{u,uv\}  ,  \{v, uv\} : uv \in \FE\} \cup \{\{u,l\} : u \in \FV, l \in \SX_u \}$ of the factor-graph edges.
The resulting \IPSLP\ formulation reads
\begin{align} \label{eq:LabelFactorsIPSLP}
& \min_{\mu,\tilde\mu}\sum_{u\in\FV}\la \theta_u,\mu_u\ra + \sum_{uv\in\FE}\la \theta_{uv},\mu_{uv}\ra + \sum_{s\in\SL}\la\tilde\theta_s,\tilde\mu_s\ra\\
& \begin{array}{ll}
   \mu \in \FL_{\FG} \\
   \tilde\mu_{s} \in \conv(X_s), & s \in \SL\\
   \mu_{u}(s) = \tilde\mu_{s}(u), & s \in X_u\,.
  \end{array} \nonumber
\end{align}
Here we introduced 
(i) auxiliary variables $\tilde\mu_{s}(u)$ for all variables $\mu_u(s)$ and 
(ii) auxiliary node costs $\tilde\theta_{s} \equiv 0$ $\forall s \in \SL$, which may take other values in course of optimization. 
Factors associated with the vectors $\mu_u$ and $\mu_{uv}$ correspond to the nodes and edges of the graph $\FG$ (node- and edge-factors), as in Example~\ref{example:LocalPolytope} and are coupled in the same way. Additionally, factors associated with the vectors~$\tilde\mu_s$ ensure that the label $s$ can be taken at most once. These label-factors are coupled with node-factors (last line in~\eqref{eq:LabelFactorsIPSLP}).
\end{example}


\begin{example}[Multicut]\label{example:multicut}
        The multicut problem (also known as correlation clustering) for an undirected weighted graph $\FG=(\FV,\FE)$
is to find a partition $(\Pi_1,\ldots,\Pi_k)$, $\Pi_i\subseteq \FV$, $\FV=\dot{\cup}_{i=1}^{k}\Pi_i$ of the graph vertexes, such that the total cost of edges connecting different components is minimized.
The number $k$ of components is not fixed but is determined by the algorithm.
See Fig.~\ref{fig:Multicut} for an illustration.
        Although the problem has numerous applications in computer vision~\cite{BreakAndConquerAlushGoldberger,ImageSegmentationClosednessAndres,ConnectomicsAndres,PlanarCorrelationClusteringYarkony} and beyond~\cite{DeduplicationMulticut,ClusteringQueryRefinement,ClusteringSparseGraphs,CorrelationClusteringMapReduce}, there is no scalable solver, which could provide optimality bounds.
Existing methods are either efficient primal heuristics~\cite{MLforCoreferenceResolution,ImprovingMLforCoreferenceResolution,ClusteringAggregation,ConversationDisentanglementCC,BoundingAndComparingCorrelationClustering,FusionMoveCC,CGC} or combinatorial branch-and-bound/branch-and-cut/column generation algorithms, based on off-the-shelf LP solvers~\cite{KappesMulticut,HigherOrderSegmentationByMulticuts,SungwoongHigherOrderCorrelationClustering,PlanarCorrelationClusteringYarkony}.
Move-making algorithms do not provide lower bounds, hence, one cannot judge their solution quality or employ them in branch-and-bound procedures. 
Off-the-shelf LP solvers on the other hand scale super-linearly, limiting their application in large-scale problems.

Instead of directly optimizing over partitions (which has many symmetries making optimization difficult in a linear programming setting), we follow~\cite{ChopraMulticut} and formulate the problem in the edge domain.
Let $\theta_e$, $e\in\FE$ denote the cost of graph edges and let $\FC$ be the set of all cycles of the graph~$\FG$. 
Each edge that belongs to different components is called a \emph{cut edge}.
The multicut problem reads
\begin{align}\label{equ:multicut-def}
 \hspace{-10pt}\min_{x_{\FE}\in\{0,1\}^{|\FE|}}\sum_{e\in\FE}\theta_{e}x_{e}\,,\ \ \text{s.t.}\
 \forall \FC\ \forall e'\in \FC \colon \sum_{e\in \FC\backslash\{e'\}}\hspace{-5pt}x_{e} \ge x_{e'}  \,.
\end{align}
Here $x_e=1$ signifies a cut edge and the inequalities force each cycle to have none or at least two cut edges.
        The formulation~\eqref{equ:multicut-def} has exponentially many constraints. However, it is well-known that it is sufficient to consider only chordless cycles~\cite{ChopraMulticut} in place of the set $\FC$ in~\eqref{equ:multicut-def}. Moreover, the graph can be triangulated by adding additional edges with zero weights and therefore the set of chordless cycles reduces to edge triples. Such triangulation is refered to as {\em chordal completion} in the literature~\cite{GareyIntractabilityNP}.  The number of triples is cubic, which is still too large for practical efficiency and therefore violated constraints are typically added to the problem iteratively in a cutting plane manner~\cite{KappesMulticut,HigherOrderSegmentationByMulticuts}. To simplify the description, we will ignore this fact below and consider all these cycles at once.
        Assuming a triangulated graph and redefining $\FC$ as the set of all chordless cycles (triples) we consider the following \IPSLP\ relaxation of the multicut problem~\footnote{One can show that this relaxation coincides with the standard LP relaxation for the multicut problem~\cite{ChopraMulticut}}:
\begin{align}\label{eq:multicut-objective}
 & \min_{\mu,\tilde\mu}\sum_{e\in\FE}\theta_{e}\mu_{e} + \sum_{c\in \FC}\sum_{e\in c}\tilde\theta_{e,c}\tilde\mu_{e,c}\,,\quad \text{s.t.}\\
 &{\small
 \left\{
 \begin{array}{l}
  \mu_{e}\in \conv\{\{0,1\}\}= [0,1],\ e\in\FE\\
  \forall c\in \FC,\ e\in c\colon \\
  \tilde\mu_{c}:=(\tilde\mu_{e,c})_{e\in c}\in\conv\{\{0,1\}^3|\ \forall e'\in\ c\colon\hspace{-5pt} \sum\limits_{e\in c\backslash\{e'\}}\hspace{-9pt}\tilde\mu_{e,c} \ge \tilde\mu_{e',c}\}\\
  \hspace{10pt}\equiv\conv\{\{0,0,0\},\{0,1,1\},\{1,0,1\},\{1,1,0\},\{1,1,1\}\}\\
  \mu_e = \tilde\mu_{e,c}
 \end{array}
\right.
}\label{eq:multicut-constraints}
\end{align}
For the sake of notation we shortened a feasible set definition $\tilde\mu\in\conv\{\mu'\in\{0,1\}^n\colon \text{constraints on}\ \mu'\}$ to $\tilde\mu\in\conv\{\{0,1\}^n\colon \text{constraints on}\ \tilde\mu\}$.
Here $\mu_e$ is the relaxed (potentially non-integer) variable corresponding to $x_e$. Variable $\tilde\mu_{e,c}$ is a copy of $\mu_e$, which corresponds to the cycle~$c$. 
Therefore, each $\mu_e$ gets as many copies $\tilde \mu_{e,c}$, as many chordless cycles $c$ contain the edge $e$.
For each cycle the set of binary vectors satisfying the cycle inequality is considered. For a cycle with $3$ edges this set can be written explicitly as in~\eqref{eq:multicut-constraints}.
Along with copies of $\mu_e$, $e\in\FE$ we copy the corresponding cost $\theta_e$ and create auxiliary costs $\tilde\theta_{e,c}\equiv 0$ for each cycle $c$ containing the edge~$e$.
During optimization, the cost $\theta_e$ will be redistributed between  $\theta_e$ itself and its copies $\tilde\theta_{e,c}$, $c\in \FC$.
The factors of the \IPSLP\ are associated with each edge (variable $\mu_e$) and each chordless cycle (variable~$\tilde\mu_{c}$). Coupling constraints connect edge-factors with those cycle-factors, which contain the corresponding edge (see the last constraint in~\eqref{eq:multicut-constraints}).
An in-depth discussion of message passing for the multicut problem with tighter relaxations can be found in~\cite{MulticutMessagePassing}.
\end{example}

\definecolor{mycolor1}{HTML}{009900}
    \definecolor{mycolor2}{HTML}{990000}
    \definecolor{green}{HTML}{009900}
    \definecolor{reddy}{HTML}{990000}
    \usetikzlibrary{decorations.pathmorphing}
	\tikzstyle{cut-edge}=[densely dotted,thick,reddy]
    \tikzstyle{vertex}=[circle, draw, fill=white, inner sep=0pt, minimum width=1ex]
    \tikzset{every picture/.append style={baseline,scale=1.1}}

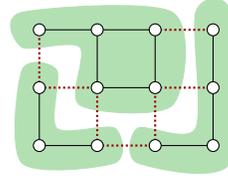
\begin{SCfigure}
   \centering
        \begin{tikzpicture}[scale=0.7]
    \draw[draw=mycolor1!30, fill=mycolor1!30] plot[smooth cycle, tension=0.5] coordinates
        {(-0.3, 2.3) (2.3, 2.3) (2.3, 0.7) (0.7, 0.7) (0.7, 1.7) (-0.3, 1.7)};
    \draw[draw=mycolor1!30, fill=mycolor1!30] plot[smooth cycle, tension=0.5] coordinates
        {(-0.3, -0.3) (-0.3, 1.3) (0.3, 1.3) (0.3, 0.3) (1.3, 0.3) (1.3, -0.3)};
    \draw[draw=mycolor1!30, fill=mycolor1!30] plot[smooth cycle, tension=0.5] coordinates
        {(1.7, -0.3) (1.7, 0.3) (2.7, 0.3) (2.7, 2.3) (3.3, 2.3) (3.3, -0.3)};
	\draw (0, 0) -- (0, 1);
	\draw (0, 0) -- (1, 0);
	\draw (0, 2) -- (1, 2);
	\draw (1, 1) -- (1, 2);
	\draw (1, 1) -- (2, 1);
	\draw (1, 2) -- (2, 2);
	\draw (2, 1) -- (2, 2);
	\draw (2, 0) -- (3, 0);
	\draw (3, 0) -- (3, 1);
	\draw (3, 1) -- (3, 2);
	\draw[style=cut-edge] (0, 1) -- (0, 2);
	\draw[style=cut-edge] (0, 1) -- (1, 1);
	\draw[style=cut-edge] (1, 0) -- (1, 1);
	\draw[style=cut-edge] (1, 0) -- (2, 0);
	\draw[style=cut-edge] (2, 0) -- (2, 1);
	\draw[style=cut-edge] (2, 1) -- (3, 1);
	\draw[style=cut-edge] (2, 2) -- (3, 2);
	\node[style=vertex] at (0, 0) {};
	\node[style=vertex] at (1, 0) {};
	\node[style=vertex] at (0, 1) {};
	\node[style=vertex] at (0, 2) {};
	\node[style=vertex] at (1, 1) {};
	\node[style=vertex] at (1, 2) {};
	\node[style=vertex] at (2, 1) {};
	\node[style=vertex] at (2, 2) {};
	\node[style=vertex] at (2, 0) {};
	\node[style=vertex] at (3, 0) {};
	\node[style=vertex] at (3, 1) {};
	\node[style=vertex] at (3, 2) {};
   \end{tikzpicture}
        \caption{\protect\rule{0ex}{-50ex}Illustration of Example~\ref{example:multicut}.
A multicut of a graph induced by three connected components $\Pi_1,\Pi_2,\Pi_3$ (\textcolor{mycolor1}{green}). \textcolor{reddy}{Red dotted} edges indicate cut edges $x_e = 1$.
}
   \label{fig:Multicut}
\end{SCfigure}

\section{Dual Problem and Admissible Messages}
\label{sec:DualProblem}


Since our technique can be seen as a dual ascent, we will not optimize the primal problem~\eqref{eq:FactorGraphPrimal} directly, but instead maximize its dual lower bound.

\myparagraph{Dual IPS-LP}
The Lagrangean dual to~\eqref{eq:FactorGraphPrimal} w.r.t. the coupling constraints reads
\begin{equation}
   \label{eq:DualLowerBound}
   \begin{array}{ll}
      \max_{\phi} & D(\phi) := \sum_{i \in \BF} \min_{x_i \in \SX_i} \la\theta^{\phi}_i, x_i \ra \\
      \text{s.t.} & \theta^{\phi}_i := \theta_i + \sum_{j : ij \in \BE} A_{(i,j)}^\T \phi_{(i,j)} \quad \forall i \in \BF\\
               & \phi_{(i,j)} = -\phi_{(j,i)} \quad \forall ij \in \BE
\end{array}
\end{equation}
Here $\phi_{(i,j)}\in\BR^{K}$ for $A_{(i,j)}\in\{0,1\}^{K\times d_i}$ for some $K \in \N$. The function $D(\phi)$ is called \emph{lower bound} and is concave in $\phi$.
The modified primal costs $\theta^{\phi}$ are called \emph{reparametrizations} of the potentials~$\theta$. We have duplicated the dual variables by introducing $\phi_{(i,j)}: = - \phi_{(j,i)}$ to symmetrize notation. In practice, only one copy is stored and the other is computed on the fly. 
Note that in this doubled notation the reparametrized node and edge potentials of the CRF from Example~\ref{example:LocalPolytope} read
\begin{equation*}
\begin{array}{l}
 \theta_u^{\phi}(x_u) = \theta_u(x_u) +\sum_{v\colon uv\in\FE}\phi_{u,uv}(x_u)\\
 \theta_{uv}^{\phi}(x_u,x_v) = \theta_{uv}(x_u,x_v) + \phi_{uv,v}(x_v) + \phi_{uv,u}(x_u)\\
  \phi_{u,uv} = - \phi_{uv,u}
\end{array}
\end{equation*}

It is well-known for CRFs that cost of feasible solutions are invariant under reparametrization.
We generalize this to the \IPSLP-case.
\begin{proposition}\label{prop:repa-invariant}
$\sum_{i\in\BF} \la \theta_i,\mu_i\ra = \sum_{i\in\BF} \la \theta^{\phi}_i,\mu_i\ra$, whenever $\mu_1,\ldots,\mu_k$ obey the coupling constraints.
\end{proposition}

\myparagraph{Admissible Messages}
While Proposition~\ref{prop:repa-invariant} guarantees that the primal problem is invariant under reparametrizations, the dual lower bound $D(\phi)$ is not. Our goal is to find $\phi$ such that $D(\phi)$ is maximal. By linear programming duality, $D(\phi)$ will then be equal to the optimal value of the primal~\eqref{eq:FactorGraphPrimal}.

First we will consider an elementary step of our future algorithm  and show that it is non-decreasing in the dual objective. This property will ensure the monotonicity of the whole algorithm. Let~$\theta^{\phi}$ be any reparametrization of the problem and $D(\phi)$ be the corresponding dual value. Let us consider changing the reparametrization of a factor $i$ by a vector $\Delta$ with the only non-zero components~$\Delta_{(i,j)}$ and~$\Delta_{(j,i)}$ . This will change reparametrization of the coupled factors $j$ (such that $ij\in\BE$) due to $\Delta_{(i,j)}=-\Delta_{(j,i)}$. The lemma below states properties of $\Delta_{(i,j)}$ which are sufficient to guarantee improvement of the corresponding dual value $D(\phi+\Delta)$:
\begin{lemma}[Monotonicity Condition]\label{prop:dual-monotonicity}
Let $ij\in\BE$ be a pair of factors related  by the coupling constraints and $\phi_{(i,j)}$ be a corresponding dual vector. 
Let $x_i^* \in \argmin\limits_{x_i\in\SX_i} \la \theta^{\phi}_i, x_i \ra$ and $\Delta_{(i,j)}$ satisfy
\begin{equation}\label{equ:allowed-phi}
\Delta_{(i,j)}(s) 
\begin{cases}
\ge 0, & \nu(s) =1\\
\le 0, & \nu(s) =0
\end{cases},\ \text{where}\ \nu:=A_{(i,j)} x^*_i\,.
\end{equation}
Then $x_i^* \in \argmin\limits_{x_i\in\SX_i} \la \theta^{\phi+\Delta}_i, x_i \ra$
implies ${D(\phi) \le D(\phi+\Delta)}$.
\end{lemma}

\begin{example}\label{example:addmissinle-dual} Let us apply Lemma~\ref{prop:dual-monotonicity} to Example~\ref{example:LocalPolytope}.
 Let $ij$ correspond to $\{u,uv\}$, where $u\in\FV$ is some node and $uv\in\FE$ is any of its incident edges.
 Then $x_i^*$ corresponds to a locally optimal label $x^*_u\in\arg\min_{s\in X_u}\theta_u(s)$ and $\nu(s)=\llbracket s= x^*_u\rrbracket$. Therefore we may assign $\Delta_{u,uv}(s)$ to any value from $[0,\theta_u(x^*_u) - \theta_u(s)]$. This assures that \eqref{equ:allowed-phi} is fulfilled and $x^*_u$ remains a locally optimal label after reparametrization even if there are multiple optima in $X_u$.
\end{example}

 Lemma~\ref{prop:dual-monotonicity} can be straightforwardly generalized to the case, when more than two factors must be reparametrized simultaneously. In terms of Example~\ref{example:LocalPolytope} this may correspond to the situation when a graph node sends messages to several incident edges at once: 
 
 \begin{definition}
Let ${i\in\BF}$ be a factor and ${J = \{j_1,\ldots,j_l\} \subseteq \SN_{\BG}(i)}$ be a subset of its neighbors. 
Let $\theta_i^{\Delta}:=\theta_i+\sum_{j \in J} A^\T_{(i,j)} \Delta_{(i,j)}$, $\Delta_{(i,j)}(=-\Delta_{(j,i)})$ satisfies~\eqref{equ:allowed-phi} for all $j\in J$ and all other coordinates of $\Delta$ are zero. 
If there exists $x_i^* \in \argmin\nolimits_{x_i\in\SX_i} \la \theta_i, x_i \ra$ such that ${x_i^* \in \argmin\nolimits_{x_i\in\SX_i} \la \theta^{\Delta}_i, x_i \ra}$, the dual vector $\Delta$ is called {\em admissible}. The set of admissible vectors is denoted by $AD(\theta_i,x^*_i,J)$.
\end{definition}

\begin{lemma}\label{prop:dual-monotonicity-extended}
 \hspace{-1pt}Let $\Delta\in AD(\theta^{\phi}_i,x^*_i,J)$ then $D(\phi)\le D(\phi+\Delta)$. 
\end{lemma}

\begin{algorithm}
\label{alg:MessagePassingUpdate}
\SetAlgorithmName{Procedure}{}\\
\textbf{Input:}Factor $i\in\BF$, neighboring factors $J = \{j_1,\ldots,j_l\} \subseteq \SN_{\BG}(i)$, dual variables $\phi$
\begin{equation}\label{eq:FactorOptimization}
\hspace{-70pt}\text{Compute}\ x^*_i\in\arg\min\nolimits_{x_i\in X_i}\la \theta^{\phi},x_i\ra \hfill
\end{equation}
\begin{equation}\label{equ:delta-def}
\hspace{-13pt} \text{Choose}\ \delta \in \R^{d_i}\ \text{s.t.}\ \delta(s) \left\{ \begin{array}{ll} > 0, & x^*_i(s) = 1 \\ < 0,& x^*_i(s) = 0 \end{array} \right. 
\end{equation}

Maximize admissible messages to $J$:
\begin{equation}\label{eq:ReverseFactorOptimization}
\Delta^*_{(i,J)} \in 
\argmax\limits_{\Delta\in AD(\theta^{\phi}_i,x^*_i,J)} \la \delta,\theta_i^{\phi+\Delta} \ra 
\end{equation}\\
\textbf{Output:} $\Delta^*_{(i,J)}$.
\caption{Message-Passing Update Step. }
\end{algorithm}

\myparagraph{Message-Passing Update Step}
To maximize $D(\phi)$, we will iteratively visit all factors and adjust messages $\phi$ connected to it, monotonically increasing the lower bound~\eqref{eq:DualLowerBound}. Such an elementary step is defined by Procedure~\ref{alg:MessagePassingUpdate}.

Procedure~\ref{alg:MessagePassingUpdate} is defined up to the vector $\delta$, which satisfies~\eqref{equ:delta-def} (see Proc.~\ref{alg:MessagePassingUpdate}). 
Usually, $\delta(s)=\left\{ \begin{array}{rl} 1, & x^*_i(s) = 1 \\ -1 ,& x^*_i(s) = 0 \end{array}\right.$ is a good choice.
Although different $\delta$ may result in different efficiency of our framework, fulfillment of~\eqref{equ:delta-def} is sufficient to prove its convergence properties.

The reparametrization adjustment problem~\eqref{eq:ReverseFactorOptimization} serves an intuitive goal to move as much slack as possible from the factor $i$ to its neighbors $J$. For example, for the setting of Example~\ref{example:addmissinle-dual} its solution reads $\Delta_{u,uv}(s) = \theta^{\phi}_u(x^*_u) - \theta^{\phi}_u(s)$.
Depending on the selected $\delta$ it might correspond to maximization of the dual objective in the direction defined by admissible reparametrizations.
Although maximization~\eqref{eq:ReverseFactorOptimization} is not necessary to prove convergence of our method (as we show below, only a feasible solution of~\eqref{eq:ReverseFactorOptimization} is required for the proof),
(i) it leads to faster convergence; (ii) for the case of CRFs (as in Example~\ref{example:LocalPolytope}) it makes our method equivalent to well established techniques like TRW-S~\cite{TRWSKolmogorov} and SRMP~\cite{SRMPKolmogorov}, as shown in Section~\ref{sec:parameter-selection}.

The following proposition states that the elementary update step defined by Procedure~\ref{alg:MessagePassingUpdate} can be performed efficiently. That is, the size of the reparametrization adjustment problem~\eqref{eq:ReverseFactorOptimization} grows linearly with the size of the factor $i$ and its attached messages:
\begin{proposition}\label{prop:repa-adjustment-is-LP}
Let $\conv(\SX_i) = \{ \mu_i : A_i \mu_i \leq b_i\}$ with $A \in \{0,1\}^{n \times m}$.
Let the messages in problem~\eqref{eq:ReverseFactorOptimization} have size $n_1,\ldots,n_{\abs{J}}$.
        Then~\eqref{eq:ReverseFactorOptimization} is a linear program with $O(n + n_1 + \ldots + n_{\abs{J}})$ variables and $O(m + n_1 + \ldots + n_{\abs{J}})$ constraints.
\end{proposition}



\section{Message Passing Algorithm}\label{sec:MessagePassingAlgorithm}

Now we combine message passing updates into Algorithm~\ref{alg:AbstractMessagePassing}.
It visits every node of the factor graph and performs the following two operations:
(i) {\bf Receive Messages}, when messages are received from a subset of neighboring factors,  
and (ii) {\bf Send Messages}, when messages to some neighboring factors are computed and reweighted by $\omega$.
Distribution of weights $\omega$ may influence the efficiency of Algorithm~\ref{alg:AbstractMessagePassing} just like it influences the efficiency of message passing for CRFs (see~\cite{SRMPKolmogorov}). We provide typical settings in Section~\ref{sec:parameter-selection}. 
Usually, factors are traversed in some given a-priori order alternately in forward and backward direction, as done in TRW-S~\cite{TRWSKolmogorov} and SRMP~\cite{SRMPKolmogorov}. We refer to~\cite{SRMPKolmogorov} for a motivation for such a schedule of computations.

\begin{algorithm}
\label{alg:AbstractMessagePassing}
\For{$i \in \BF$ in some order} {
\underline{\textbf{Receive Messages:}}\\
Choose a subset of connected factors $J_{receive} \subseteq \SN_{\BG}(i)$\\
\For{$j \in J_{receive}$} {
Compute $\Delta^*_{(j,\{i\})}$ with Procedure~\ref{alg:MessagePassingUpdate}. \\
Set \quad $\phi = \phi + \Delta^*_{(j,\{i\})}$. \\
}
\ \\
\underline{\textbf{Send Messages:}}\\
Choose partition $J_1 \dot{\cup} \ldots \dot{\cup} J_l \subseteq \SN_{\BG}(i)$. \label{linenum:alg2:partition}\\
\For{$J \in \{J_1,\ldots,J_l\}$} {
Compute $\Delta^*_{(i,J)}$ with Procedure~\ref{alg:MessagePassingUpdate}.
}
Choose weights $\omega_{J_1},\ldots,\omega_{J_l} \geq 0$ such that $ \omega_{J_1} + \ldots +\omega_{J_l} \leq 1$.\\
\For{$J \in \{J_1,\ldots,J_l\}$} {
   Set $\phi = \phi + \omega_{J} \Delta^*_{(i,J)}$. \\
}
}
\caption{One Iteration of Message-Passing}
\end{algorithm}
We will discuss parameters of Algorithm~\ref{alg:AbstractMessagePassing} (factor partitioning $\{J_i\}$, weights $w_{J_i}$) right after the theorem stating monotonicity for any choice of parameters.

\begin{theorem}\label{thm:Algorithm-is-Dual-Monotone}
   Algorithm~\ref{alg:AbstractMessagePassing} monotonically increases the dual lower bound~\eqref{eq:DualLowerBound}.
\end{theorem}

\subsection{Parameter Selection for Algorithm~\ref{alg:AbstractMessagePassing}}\label{sec:parameter-selection}


There are the following free parameters in Algorithm~\ref{alg:AbstractMessagePassing}: 
(i) The order of traversing factors of $\BF$; (ii) for each factor the neighboring factors from which to receive messages $J_{receive} \subseteq \SN_{\BG}(i)$;
(iii) the partition $J_1 \dot{\cup} \ldots \dot{\cup} J_l \subseteq \SN_{\BG}(i)$  of factors to send messages to
and (iv) the associated weights $\omega_{J_1},\ldots, \omega_{J_l}$ for messages.

Although for any choice of these parameters Algorithm~\ref{alg:AbstractMessagePassing} monotonically increases the dual lower bound (as stated by Theorem~\ref{thm:Algorithm-is-Dual-Monotone}), its efficiency may significantly depend on their values. Below, we will describe the parameters for Examples~\ref{example:LocalPolytope}-\ref{example:multicut}, which we found the most efficient empirically. Additionally, in the supplement we discuss parameters, which turn our algorithm into existing message passing solvers for CRFs (as in Example~\ref{example:LocalPolytope}).

Sending a message by some factor automatically implies receiving this message by another, coupled factor. Therefore, usually there is no need to go over all factors in Algorithm~\ref{alg:AbstractMessagePassing}. 
It is usually sufficient to guarantee that all coupling constraints are updated by Procedure~\ref{alg:MessagePassingUpdate}.
Formally, we can always exclude processing some factors by setting $J_{receive}$ and $J_i$, $i=1,\dots,l$ to the empty set. Instead, we will explicitly specify, which factors are processed in the loop of Algorithm~\ref{alg:AbstractMessagePassing} in the examples below.

\myparagraph{Parameters for Example~\ref{example:LocalPolytope}, MAP-inference in CRFs.}
Pairwise CRFs have the specific feature that node factors are coupled with edge factors only. This implies that processing only node factors in Algorithm~\ref{alg:AbstractMessagePassing} is sufficient. Below, we describe parameters, which turn Algorithm~\ref{alg:AbstractMessagePassing} into SRMP~\cite{SRMPKolmogorov} (which is up to details of implementation equivalent to TRW-S~\cite{TRWSKolmogorov} for pairwise CRFs). Other settings, given in the supplement, may turn it to other popular message passing techniques like MPLP~\cite{MPLP} or min-sum diffusion~\cite{schlesinger2011diffusion}.

We order node factors and process them according to this ordering. The ordering naturally defines the sets of incoming $\FE_u^{+}$ and outgoing $\FE_u^{-}$ edges for each node $u\in\FV$. Here $uv\in\FE$ is {\em incoming} for $u$ if $v < u$ and {\em outgoing} if $v > u$. Each node $u\in\FV$ receives messages from all incoming edges, which is $J_{receive} = \SN_{\BG}(u)=\FE_u^{+}$. The messages are send to all outgoing edges, each edge $uv\in\FE$ in the partition in line~\ref{linenum:alg2:partition} of Algorithm~\ref{alg:AbstractMessagePassing} is represented by a separate set. That is, the partition reads $\dot{\cup}_{e\in \FE_u^{-}}\{e\}$. Weights are distributed uniformly and equal to $w_e=\{\frac{1}{\max\{|\FE_u^{-}|,|\FE_u^{+}|\}}\}$, $e\in \FE_u^{-}$. 
After each outer iteration, when all nodes were processed, the ordering is reversed and the process repeats. We refer to~\cite{SRMPKolmogorov} for substantiation of these parameters.

\myparagraph{Parameters for Example~\ref{example:graph-matching}, Graph Matching.}
Additionally to the node and edge factors, the corresponding \IPSLP\ has also label factors~\eqref{eq:LabelFactorsIPSLP}. 
To this end all node factors are ordered, as in Example~\ref{example:LocalPolytope}. Each node factor $u\in\FV$ receives messages from all incoming edge factors and label factors $J_{receive}(u) = \FE_u^{+} \cup X_u$ and sends them to all outgoing edges {\em and} label factors. The corresponding partition reads $\dot{\cup}_{f\in \SN_{\BG}(u)\backslash \FE_u^{+}}\{f\}\ \dcup X_u$. The weights are distributed uniformly with $w_f=\{\frac{1}{1+\max\{|\FE_u^{-}|,|\FE_u^{+}|\}}\}$. 
The label factors are processed after all node factors were visited.
Each label factor receives messages from all connected node factors and send messages back as well: $J_{receive}(s) = \{u \in \FV: s \in X_u\}$. 
We use the same single set for sending messages, i.e. $J_1 = J_{receive}$.
After each iteration we reverse the factor order.

\myparagraph{Parameters for Example~\ref{example:multicut}, Multicut.}
Similarly to Example~\ref{example:LocalPolytope}, it is sufficient to go only over all edge factors in the loop of Algorithm~\ref{alg:AbstractMessagePassing}, since each coupling constraint contains exactly one cycle and one edge factor.
Each edge factor $e$ receives messages from all coupled cycle factors  $J_{receive}=\SN_{\BG}(\{c \in C : e \in c\})$ and sends them to the same factors.
As in Example~\ref{example:LocalPolytope}, each cycle factor forms a trivial set in the partition in line~\ref{linenum:alg2:partition}  of Algorithm~\ref{alg:AbstractMessagePassing}, the partition reads $\dot{\cup}_{c \in C : e \in c}\{c\}$. Weights are distributed uniformly with $w_e=\frac{1}{\abs{c \in C : e \in c}}$. 
After each iteration the processing order of factors is reversed.


\subsection{Obtaining Integer Solution}
Eventually we want to obtain a primal solution $x \in X$ of~\eqref{eq:FactorGraphPrimal}, not a reparametrization $\theta^{\phi}$.
We are not aware of any rounding technique which would work equally well for all possible instances of \IPSLP\ problem. 
According to our experience, the most efficient rounding is problem specific. Below, we describe our choices for the Examples \ref{example:LocalPolytope} -- \ref{example:multicut}.
\myparagraph{Rounding for Example~\ref{example:LocalPolytope}} coincides with the one suggested in~\cite{TRWSKolmogorov}:
Assume we have already computed a primal integer solution $x^*_v$ for all $v<u$ and we want to compute~$x^*_u$. 
To this end, right before the message receiving step of Algorithm~\ref{alg:AbstractMessagePassing} for $i=u$ we assign
\begin{equation}
         x_u^* \in \argmin_{x_u} \theta_u(x_u) + \sum_{v < u: uv \in \FE} \theta_{uv}(x_u, x_v^*)\,.
\end{equation} 

\myparagraph{Rounding for Example~\ref{example:graph-matching}} is the same except that we select the best label $x_u$ among those, which have not been assigned yet, to satisfy uniqueness constraints:
\begin{equation}
        x_u^* \in \argmin_{x_u : x_v^* \neq x_u \forall v < u} \theta^{\phi}_u(x_u) + \sum_{v < u: uv \in \FE} \theta^{\phi}_{uv}(x_u, x_v^*)\,.
\end{equation} 


\myparagraph{Rounding for Example~\ref{example:multicut}.}
We use the efficient Kernighan\&Lin heuristic~\cite{KernighanLin} as implemented in~\cite{ImageMeshDecompositionLiftedMulticut}. Costs for the rounding are the reparametrized edge potentials.


\section{Fixed Points and Comparison to Subgradient Method}
\label{sec:FixedPoints}
Algorithm~\ref{alg:AbstractMessagePassing} does not necessarily converge to the optimum of~\eqref{eq:FactorGraphPrimal}.
Instead, it may get stuck in suboptimal points, similar to those correspoding to the "weak tree agreement"~\cite{TRWSKolmogorov} or "arc consistency"~\cite{Werner07} in CRFs from Example~\ref{example:LocalPolytope}.
Below we characterise these fixpoints precisely.

\begin{definition}[Marginal Consistency]
Given a reparametrization $\theta^{\phi}$, let for each factor $i \in \BF$ a non-empty set $\BS_i \subseteq \argmin_{x_i \in \SX_i} \la \theta^{\phi}, x_i \ra$, $i\in\BF$ be given.
Define $\BS = \prod_{i \in \BF} \BS_i$.
We call reparametrization $\theta^{\phi}$ {\em marginally consistent for $\BS$ on $ij \in \BE$} if 
\begin{equation}
   A_{(i,j)} \left( \BS_i \right) = A_{(j,i)} \left(\BS_j\right)\,.
\end{equation}
If $\theta^{\phi}$ is marginally consistent for $\BS$ on all $ij \in \BE$, we call $\theta^{\phi}$ {\em marginally consistent for $\BS$}.
\end{definition}

Note that marginal consistency is necessary, but not sufficient for optimality of the relaxation~\eqref{eq:FactorGraphPrimal}. 
This can be seen in the case of CRFs (Example~\ref{example:LocalPolytope}), where it exactly corresponds to arc-consistency.
The latter is only necessary, but not sufficient for optimality~\cite{Werner07}.

\begin{theorem}
\label{thm:MarginalConsistencyFixPoint}
If $\theta^{\phi}$ is marginally consistent, the dual lower bound $D(\phi)$ cannot be improved by Algorithm~\ref{alg:AbstractMessagePassing}.
\end{theorem}

\myparagraph{Comparison to Subgradient Method.}
Decomposition~\IPSLP\ and more general ones can be solved via the subgradient method~\cite{DualDecompositionKomodakis}.
Similar to Algorithm~\ref{alg:AbstractMessagePassing}, it operates on dual variables $\phi$ and manipulates them by visiting each factor sequentially. Contrary to Algorithm~\ref{alg:AbstractMessagePassing}, subgradient algorithms converge to the optimum.
Moreover, on a per-iterations basis, computing subgradients is cheaper than using Algorithm~\ref{alg:AbstractMessagePassing}, as only~\eqref{eq:FactorOptimization} needs to be computed, while Algorithm~\ref{alg:AbstractMessagePassing} needs to solve~\eqref{eq:ReverseFactorOptimization} additionally.
However, for MAP-inference, the study~\cite{OpenGMBenchmark} has shown that subgradient-based algorithms converge much slower than message passing algorithms like TRWS~\cite{TRWSKolmogorov}.
In Section~\ref{sec:experiments} we confirm this for the graph matching problem as well.

The reason for this large empirical difference is that one iteration of the subgradient algorithm only updates those coordinates of dual variables $\phi$ that are affected by the current minimal labeling $x_i^* \in \argmin_{x_i \in \SX_i} \la \theta^{\phi}_i, x_i \ra$ (i.e. coordinates $k: (A^\T_{(i,j)} x_i^*)_k = 1$), while in Algorithm~\ref{alg:AbstractMessagePassing} \emph{all} coordinates of $\phi$ are taken into account. 
Also message passing implicitly chooses the stepsize so as to achieve monotonical convergence in Algorithm~\ref{alg:MessagePassingUpdate}, while subgradient based algorithms must rely on some stepsize rule that may either make too large or too small changes to the dual variables $\phi$.

\section{Experimental Evaluation}\label{sec:experiments}
Our experiments' goal is to illustrate applicability of the proposed technique, they are not an exhaustive evaluation.
The presented algorithms are only basic variants, which can be further improved and tuned to the considered problems.
Both issues are addressed in the specialized studies~\cite{GraphMatchingMessagePassing,MulticutMessagePassing} which are appended.
Still, we show that the presented basic variants are already able to surpass state-of-the-art specialized solvers on challenging datasets.
All experiments were run on a computer with a 2.2 GHz i5-5200U CPU and 8 GB RAM.

\begin{figure*}
        \begin{minipage}{0.22\textwidth}
                \centering
                \includegraphics[width=1.5\linewidth]{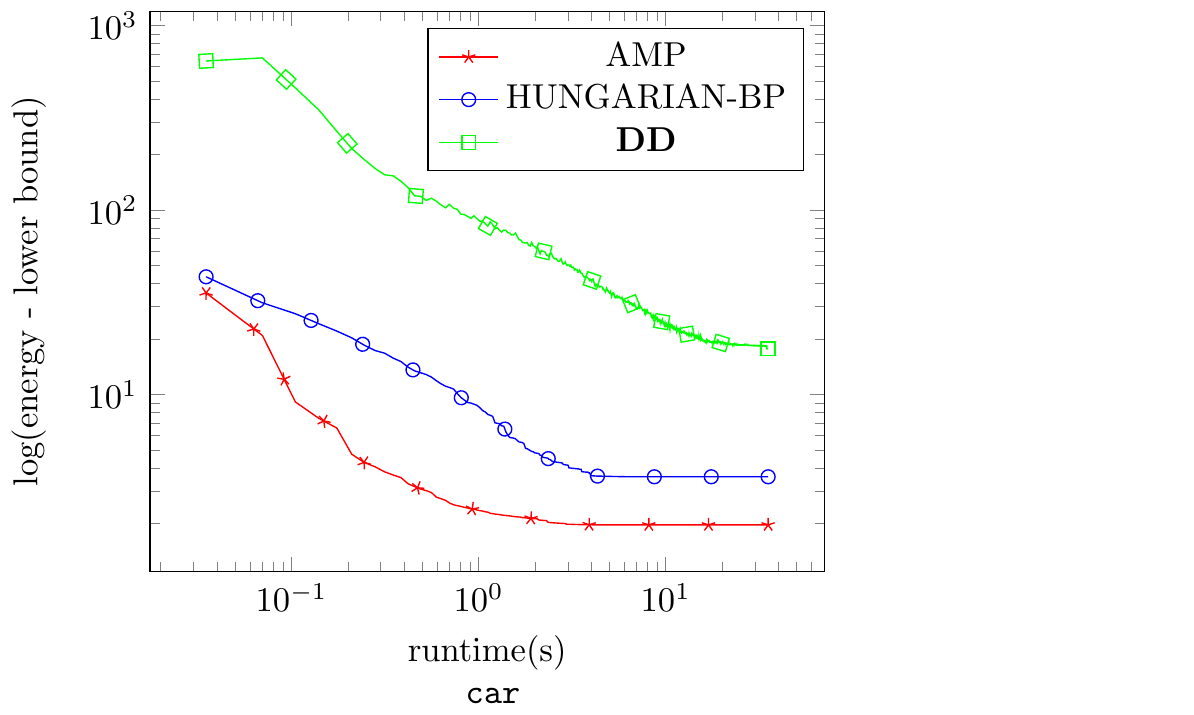}
        \end{minipage}
        \hspace{1cm}
        \begin{minipage}{0.22\textwidth}
                \centering
                \includegraphics[width=1.5\linewidth]{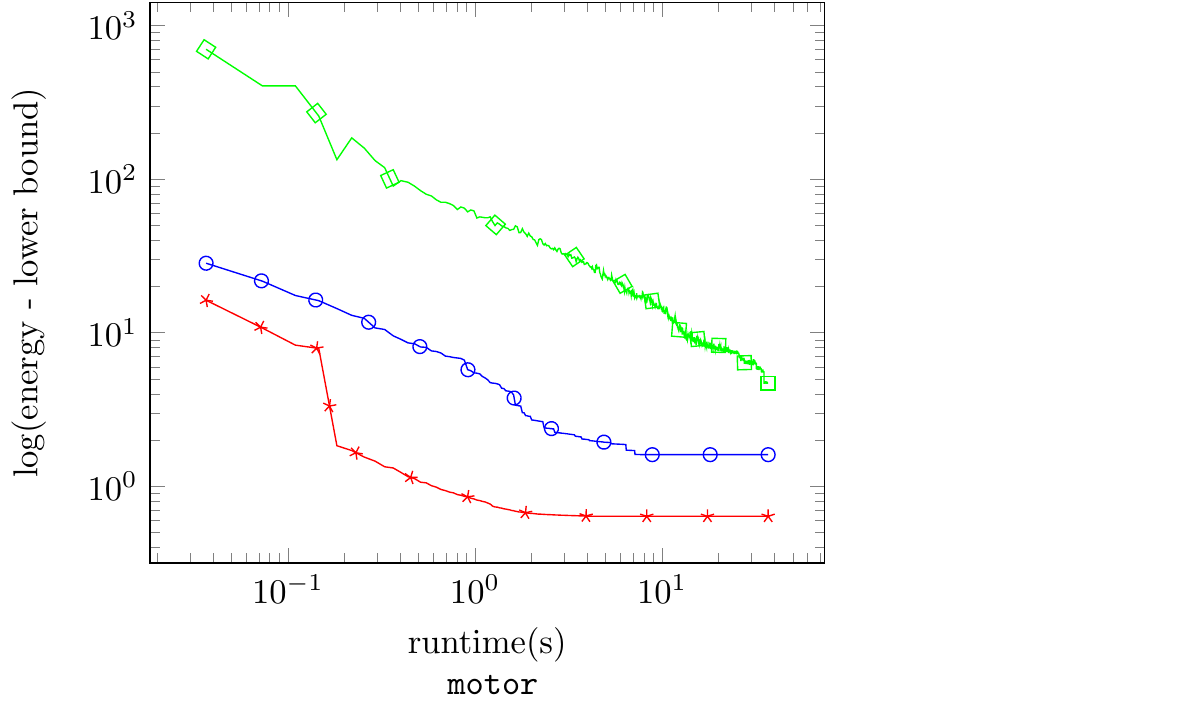}
        \end{minipage}
        \hspace{1cm}
        \begin{minipage}{0.22\textwidth}
                \centering
                \includegraphics[width=1.5\linewidth]{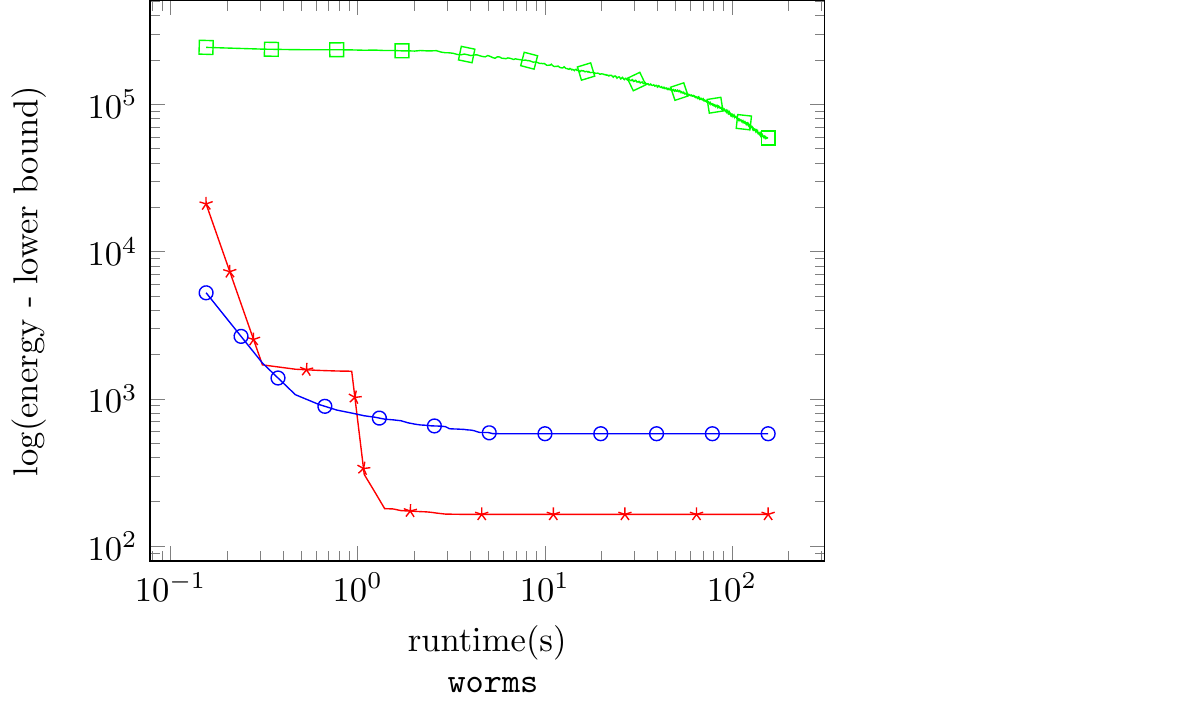}
        \end{minipage}
        \hfill
        \caption{
                Runtime plots comparing averaged $\log(\text{primal energy} - \text{dual lower bound})$ values on \texttt{car}, \texttt{motor} and \texttt{worms} graph matching datasets.
                Both axes are logarithmic. 
        }
   \label{fig:graph-matching-experiment}
\end{figure*}

\begin{figure*}
        \centering
        \begin{minipage}{0.22\textwidth}
                \centering
                \includegraphics[width=1.5\linewidth]{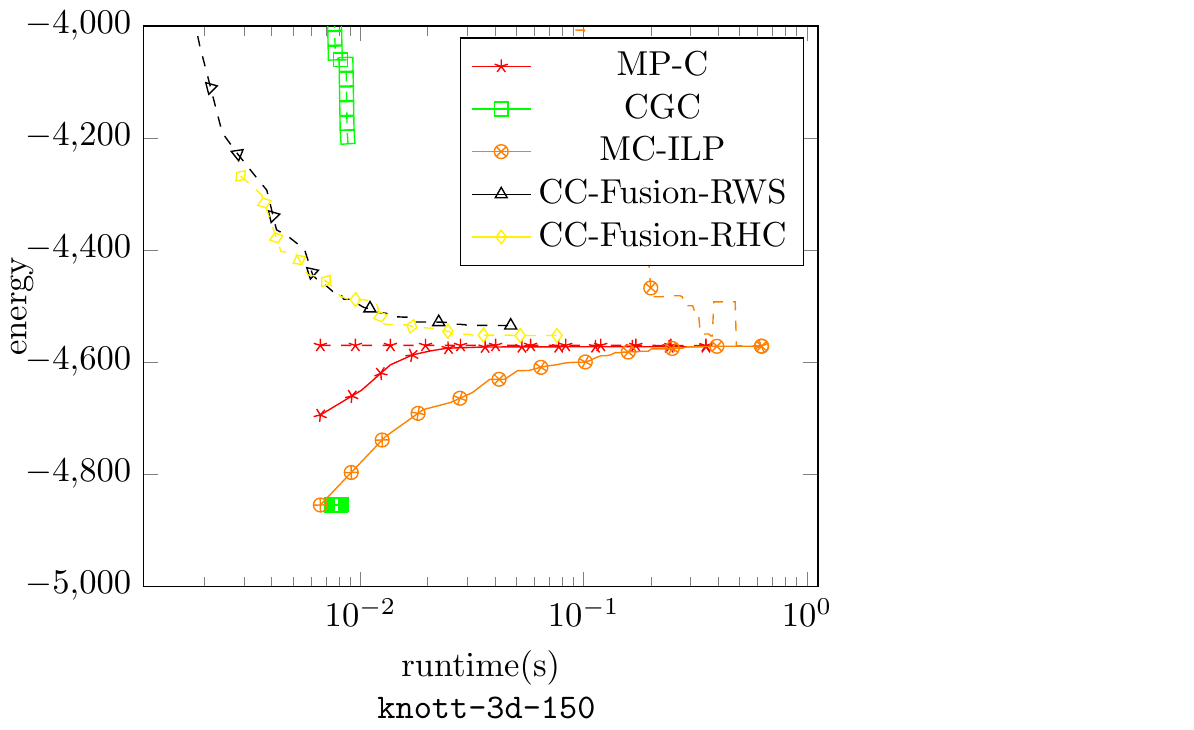}
        \end{minipage}
        \hspace{1cm}
        \begin{minipage}{0.22\textwidth}
                \centering
                \includegraphics[width=1.5\linewidth]{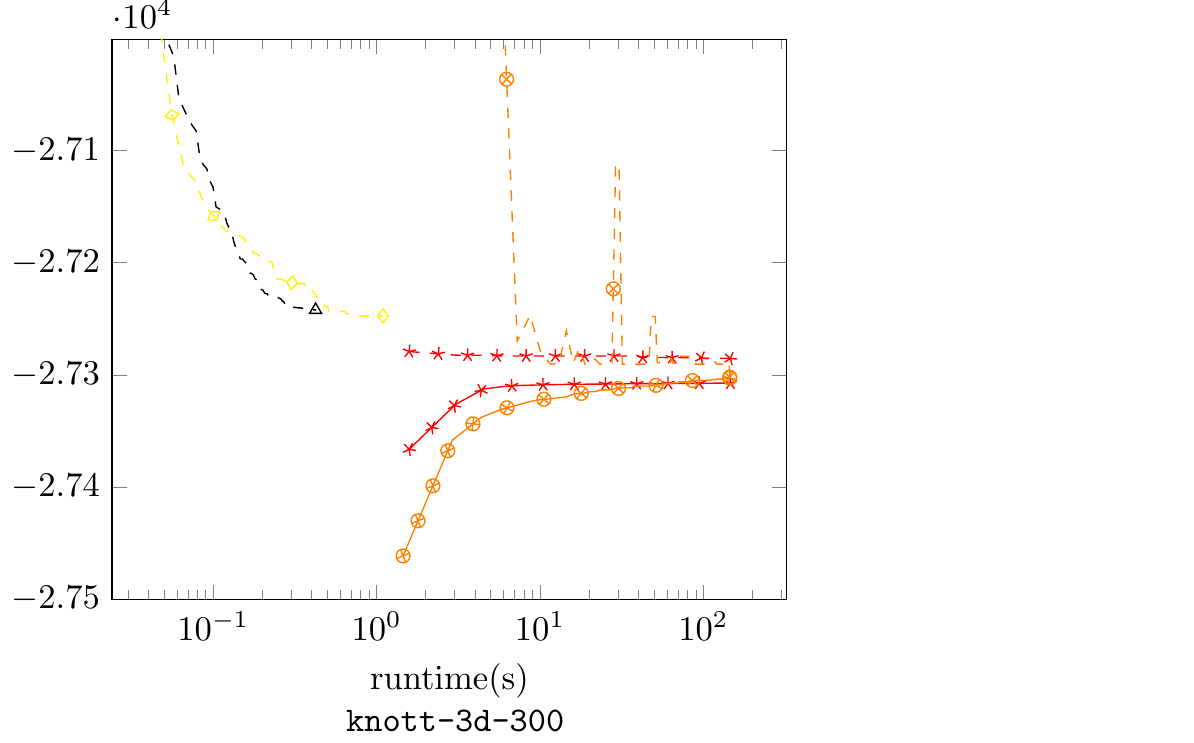}
        \end{minipage}
        \hspace{1cm}
        \begin{minipage}{0.22\textwidth}
                \centering
                \includegraphics[width=1.5\linewidth]{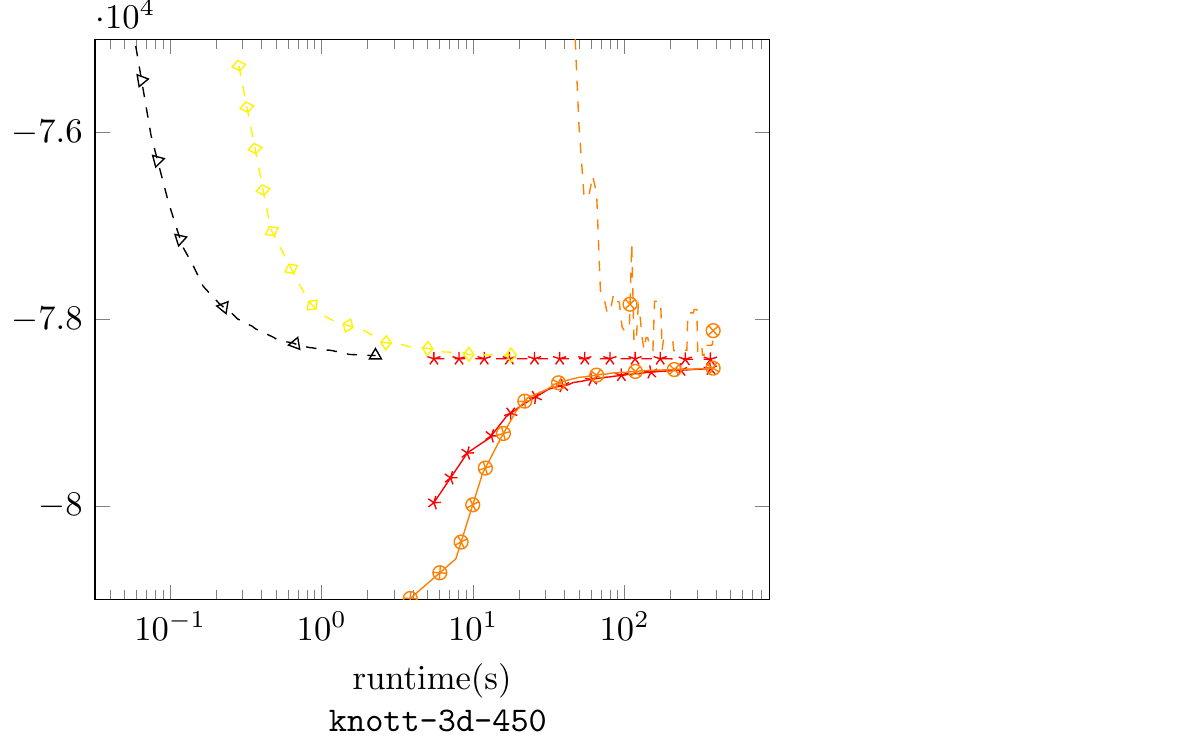}
        \end{minipage}
        \caption{
                Runtime plots comparing averaged primal/dual values on the three \texttt{knott-3d-\{150|300|450\}} multicut datasets. Values are averaged over all instances in the dataset. The x-axis are logarithmic. Continuous lines are dual lower bounds while corresponding dashed lines show primal solutions obtained by rounding.
        }
   \label{fig:multicut-experiment}
\end{figure*}

\subsection{Graph Matching}
\myparagraph{Solvers.}
We compare against two state-of-the-art algorithms: (i) the subgradient based dual decomposition solver~\cite{GraphMatchingDDTorresaniEtAl} abbreviated by \textbf{DD} and 
(ii) the recent ``hungarian belief propagation'' message passing algorithm~\cite{HungarianBP}, abbreviated as \textbf{HBP}.
While the authors of~\cite{HungarianBP} have embedded their solver in a branch-and-bound routine to produce exact solutions, we have reimplemented their message passing component but did not use branch and bound to make the comparison fair. 
Both algorithms \textbf{DD} and \textbf{HBP} outperformed alternative solvers at the time of their publication, hence we have not tested against~\cite{CombinatorialOptimizationMaxProductDuchi,CoveringTreesLowerBoundQuadraticAssignmentJarkony,SpectralTechniqueAssignmentLeordeanu,MRFSemidefiniteTorr,ProbabilisticSubgraphMatchingSchellewald,GraduatedAssignmentGold,FactorizedGraphMatching,LocalSparseMatching,IntegerFixedPointGraphMatching,RandomWalksForGraphMatching}.
We call our solver~\textbf{AMP}.
\myparagraph{Datasets.}
We selected three challenging datasets. The first two are the standard benchmark datasets \texttt{car} and \texttt{motor}, both used in~\cite{UnsupervisedLearningForGraphMatching}, containing 30 pairs of cars and 20 pairs of motorbikes with keypoints to be matched 1:1. The images are taken from the VOC PASCAL 2007 challenge~\cite{VocPascal}. Costs are computed from features as in~\cite{UnsupervisedLearningForGraphMatching}.
Instances are densely connected graphs with $20$ -- $60$ nodes.
The third one is the novel \texttt{worms} datasets~\cite{KainMueller2014active}, containing 30 problem instances coming from bioimaging. 
The problems are made of sparsely connected graphs with up to $600$ nodes and up to $1500$ labels.
To our knowledge, the \texttt{worms} dataset contains the largest graph matching instances ever considered in the literature.
For runtime plots showing averaged logarithmic primal/dual gap over all instances of each dataset see Fig.~\ref{fig:graph-matching-experiment}.
\myparagraph{Results.}
Our solver \textbf{AMP} consistently outperforms \textbf{HBP} and \textbf{DD} w.r.t. primal/dual gap and anytime performance
Most markedly on the largest \texttt{worms} dataset, the subgradient based algorithm \textbf{DD} struggles hard to decrease the primal/dual gap, while \textbf{AMP} gives reasonable results.


\subsection{Multicuts}
\paragraph{Solvers.}
We compare against state-of-the-art multicut algorithms implemented in the OpenGM~\cite{OpenGMBenchmark} library, namely (i) the branch-and-cut based solver \textbf{MC-ILP}~\cite{HigherOrderSegmentationByMulticuts} utilizing the ILP solver CPLEX~\cite{cplex}, (ii) the heuristic primal ``fusion move'' algorithm \textbf{CC-Fusion}~\cite{FusionMoveCC} with random hierarchical clustering and random watershed proposal generator, denoted by the suffixes \textbf{-RHC} and \textbf{-RWS} and (iii) the heuristic primal ``Cut, Glue \& Cut'' solver \textbf{CGC}~\cite{CGC}.
Those solvers were shown to outperform other multicut algorithms~\cite{FusionMoveCC}.
Algorithm \textbf{MC-ILP} provides both upper and lower bounds, while \textbf{CC-Fusion} and \textbf{CGC} are purely primal algorithms. 
We call our message passing solver with cycle constraints added in a cutting plane fashion \textbf{MP-C}.

\paragraph{Datasets.}
A source of large scale problems comes from electron microscopy of brain tissue, for which we wish to obtain neuron segmentation.
We have selected three datasets \texttt{knott-3d-\{150|300|450\}}
of increasingly large size~\cite{OpenGMBenchmark}, each consisting of 8 instances.
Instances have $\leq972$, $5896$ and $17074$ 
nodes and $\leq5656$, $36221$, and $107060$ 
edges respectively.
\myparagraph{Results.}
For plots showing dual bounds and primal solution objectives over time see Figure~\ref{fig:multicut-experiment}. 
Our algorithm \textbf{MP-C} combines advantages of LP-based techniques awith those of primal heuristics:
It delivers high dual lower bounds faster than \textbf{MC-ILP}.
Its has fast primal convergence speed and delivers primal solutions comparable/superior to \textbf{CGC}'s and \textbf{CC-Fusion}'s.

\section{Acknowledgments}
The authors would like to thank Vladimir Kolmogorov for helpful discussions.
This work is partially funded by the European Research Council under the European Unions Seventh Framework Programme (FP7/2007-2013)/ERC grant agreement no 616160.

{\small
\bibliographystyle{ieee}
\bibliography{../literatur}
}

\clearpage
\section{Supplementary Material}
\myparagraph{Proof of Proposition~\ref{prop:repa-invariant}}
\begin{proposition*}
$\sum_{i\in\BF} \la \theta_i,\mu_i\ra = \sum_{i\in\BF} \la \theta^{\phi}_i,\mu_i\ra$, whenever $\mu_1,\ldots,\mu_k$ obey the coupling constraints.
\end{proposition*}
\begin{proof} 
   \hfill
$
      \sum_{i \in \BF} \la \theta^{\phi}, \mu_i \ra = 
      \sum_{i \in \BF} \la \theta, \mu_i \ra + \underbrace{\sum_{ij \in \BE} \la \phi_{(i,j)}, A_{(i,j)} \mu_i \ra + \la \phi_{(j,i)}, A_{(j,i)} \mu_j \ra }_{(*)}
      =
      \sum_{i \in \BF} \la \theta, \mu_i \ra\,,
      $ \hfill
   where $(*) = 0$ due to $\phi_{(i,j)} = -\phi_{j,i)}$ and $A_{(i,j)} \mu_i = A_{(j,i)} \mu_j$.
\end{proof}

\myparagraph{Proof of Proposition~\ref{prop:repa-adjustment-is-LP}}
\begin{proposition*}
Let $\conv(\SX_i) = \{ \mu_i : A_i \mu_i \leq b_i\}$ with $A \in \{0,1\}^{n \times m}$.
Let the messages in problem~\eqref{eq:ReverseFactorOptimization} have size $n_1,\ldots,n_{\abs{J}}$.
Then~\eqref{eq:ReverseFactorOptimization} is a linear program with $O(n + n_1 + \ldots + n_{\abs{J}})$ variables and $O(m + n_1 + \ldots + n_{\abs{J}})$ constraints.
\end{proposition*}
\begin{proof}
From LP-duality we know that $\mu_i^* \in \argmin_{\mu_i : A \mu_i \leq b_i} \la c, \mu_i \ra$ iff $\exists y\geq 0 : A_i^\T y = c_i$ and $\la b_i - A_i \mu_i^* , y \ra = 0$.
Hence,~\eqref{eq:ReverseFactorOptimization} can be rewritten as
\begin{equation}
   \label{eq:ReverseFactorOptimizationLP}
\begin{array}{cl}
\max\limits_{y \geq 0, \Delta_{(i,j_1)},\ldots,\Delta_{(i,j_l)}} 
& \la \delta,\theta^{\phi+\Delta} \ra \\
\text{s.t.}
& \la b_i - A_i \mu_i^*, y \ra = 0 \\
& A_i^\T y = \theta^{\phi+\Delta} \\
& \Delta_{(i,j)}(s) 
\begin{cases}
\le 0, & \nu_i(s) =0\\
\ge 0, & \nu_i(s) =1
\end{cases}\\
  &\hfill\text{where}\ \nu_i:=A_{(i,j)} \mu^*_i
\end{array} 
\end{equation}
$\theta^{\phi + \Delta}$ is a linear expression and $\mu^*$ is constant during the computation, hence~\eqref{eq:ReverseFactorOptimizationLP} is a LP.
\end{proof}

\myparagraph{Proof of Lemma~\ref{prop:dual-monotonicity} and Lemma~\ref{prop:dual-monotonicity-extended}}
\begin{lemma*}
Let $ij\in\BE$ be a pair of factors related  by the coupling constraints and $\phi_{(i,j)}$ be a corresponding dual vector. 
Let $x_i^* \in \argmin\limits_{x_i\in\SX_i} \la \theta^{\phi}_i, x_i \ra$ and $\Delta_{(i,j)}$ satisfy
\begin{equation}\label{equ:allowed-phi}
\Delta_{(i,j)}(s) 
\begin{cases}
\ge 0, & \nu(s) =1\\
\le 0, & \nu(s) =0
\end{cases},\ \text{where}\ \nu:=A_{(i,j)} x^*_i\,.
\end{equation}
Then $x_i^* \in \argmin\limits_{x_i\in\SX_i} \la \theta^{\phi+\Delta}_i, x_i \ra$
implies ${D(\phi) \le D(\phi+\Delta)}$.
\end{lemma*}
\begin{proof}
Let $x_j^* \in \argmin_{x_j \in \SX_j} \la \theta^{\phi}_j,x_j \ra$ be a solution of~\eqref{eq:FactorOptimization} at which the dual lower bound~\eqref{eq:DualLowerBound} is attained before the update
and $x^{**}_j \in \argmin_{x_j \in \SX_j} \la \theta^{\phi} - A_{(j,i)}^\T \Delta^*_{(i,j)}, x_j \ra$ be an integral solution at which the dual lower bound is attained after $\phi$ has been updated.
Variable $x_i^*$ as chosen in~\eqref{eq:FactorOptimization} is optimal for $\theta^{\phi}$ and for $\theta^{\phi + \Delta}$ by construction.
We need to prove 
\begin{multline}
\la \theta^{\phi}_i,x_i^*\ra + \sum_{j \in J} \la \theta^{\phi}_j, x_j^*\ra \\
\leq \la \theta^{\phi}_i + \sum_{j \in J} A^\T_{(i,j)} \Delta^*_{(i,j)} ,x^*_i\ra + \sum_{j \in J} \la \theta^{\phi}_j - A^\T_{(j,i)} \Delta^*_{(i,j)} , x^{**}_j\ra \,.
\end{multline}
We shuffle all terms with variables $\Delta^*_{(i,j)}$, $j \in J$ to the right side and all other terms to the left side.
\begin{multline}
\la \theta_i^{\phi}, x^*_i - x^*_i \ra + \sum_{j \in J} \la \theta_j^{\phi}, x^*_j - x^{**}_j \ra \\
\leq \la \sum_{j \in } A^\T_{(i,j)} \Delta^*_{(i,j)}, x^*_i \ra - \sum_{j \in J} \la A^\T_{(j,i)} \Delta^*_{(i,j)}, x^{**}_j \ra
\end{multline}
All terms on the left side are smaller than zero due to the choice of $x^*_j$ being minimizers w.r.t. $\theta^{\phi}_j$.
Hence, it will be enough to prove the above inequality when assuming the left side to be zero.
We rewrite the scalar products by transposing $A^\T_{(i,j)}$ and $A^\T_{(j,i)}$. 
\begin{equation}
0 \leq \sum\nolimits_{j \in J} \left\{\la \Delta^*_{(i,j)}, A_{(i,j)} x_i^* - A_{(j,i)} x^{**}_j \ra \right\}
\end{equation}
Due to $A_{(j,i)} x^{**}_j \in \{0,1\}^{\dim(\phi_{(i,j)})}$ and $A_{(i,j)} x_i^* \in \{0,1\}^{\dim(\phi_{(i,j)})}$ by Definition~\ref{def:IPSLP} and $\Delta^*_{(i,j)} \lessgtr 0$ whenever $A_{(i,j)} x^*_i \lessgtr 0$, the result follows.
\end{proof}

\begin{lemma*}
 \hspace{-1pt}Let $\Delta\in AD(\theta^{\phi}_i,x^*_i,J)$ then $D(\phi)\le D(\phi+\Delta)$. 
\end{lemma*}
\begin{proof}
  Analoguous to the proof of Lemma~\ref{prop:dual-monotonicity}.
\end{proof}

\myparagraph{Proof of Theorem~\ref{thm:Algorithm-is-Dual-Monotone}}
\begin{theorem*}
   Algorithm~\ref{alg:AbstractMessagePassing} monotonically increasis the dual lower bound~\eqref{eq:DualLowerBound}.
\end{theorem*}
\begin{proof}
We prove that (i) the receiving messages and (ii) the sending messages step improve~\eqref{eq:DualLowerBound}.

\noindent(i) Directly apply Lemma~\ref{prop:dual-monotonicity}.
(ii) The difficulty here is that we compute descent directions from the current dual variables $\phi$ in parallel and then apply all of them simultaneously. 
  By Lemma~\ref{prop:dual-monotonicity-extended}, the send message step is non-decreasing when called for each set $J_1,\ldots,J_l$ in Algorithm~\ref{alg:AbstractMessagePassing}.
The dual lower bound $L(\phi)$ is concave, hence we apply Jensen's inequality and note that $\omega_1 + \ldots + \omega_l \leq 1$ to obtain the result.
\end{proof}

\myparagraph{Proof of Theorem~\ref{thm:MarginalConsistencyFixPoint}}
\begin{theorem*}
If $\theta^{\phi}$ is marginally consistent, the dual lower bound $D(\phi)$ cannot be improved by Algorithm~\ref{alg:AbstractMessagePassing}.
\end{theorem*}

First, we need two technical lemmata.
\begin{lemma}
\label{lemma:OptimalityAfterUpdate}
Let $X \subset \{0,1\}^n$, $A \in \{0,1\}^{K \times n}$ and $Ax \in \{0,1\}^K$ $\forall x \in X$.
Let $x^* \in X$ be given and define $\nu^* := Ax^*$.
Let $\Delta \in \R^K$ be given such that $\Delta(s) \begin{cases} \geq 0,& \nu^*(s) =1 \\ \leq 0,& \nu^*(s) = 0 \end{cases}$.
Then (i)~$x^* \in \argmin_{x \in X} \la -\Delta, A x \ra$ 
and (ii)~for $x^{**} \in \argmin_{x \in X} \la -\Delta, A x \ra$, $\nu^{**} = A x^{**}$ it holds that $\Delta(s) = 0$ whenever $\nu^*(s) \neq \nu^{**}(s)$.
\end{lemma}
\begin{proof}
Let $x \in X$ and define $\nu = A x$.
Then 
\begin{multline}
\la -\Delta, A x \ra \\
= 
\underbrace{\sum_{s: \nu^*(s) = 1 = \nu(s) } \hspace*{-0.4cm} - \Delta(s)}_{(*)}
+
\underbrace{\sum_{s: \nu(s) = 1 > 0 = \nu^*(s) } \hspace*{-0.6cm} - \Delta(s)}_{(**)} \\
\geq 
\underbrace{\sum_{s : \nu^*(s) = 1}-\Delta(s) }_{(***)} \\
= \la -\Delta, A x^* \ra
\end{multline}
because $(*) \geq (***)$ due to $\Delta(s) \geq 0$ for $\nu^*(s) = 1$ and $(**) \geq 0$ due to $\Delta(s) \leq 0$ for $\nu^*(s) = 0$.
This proves (i) and (ii) is proven by observing that $(**) = 0$ and $(*) = (***)$ must also hold.
\end{proof}

\begin{lemma}
\label{lemma:UpdateSignsMultipleMinima}
Let $x_i^*,x_i^{**} \in \argmin_{x_i \in \SX_i} \la \theta^{\phi}, x_i \ra$ be two solutions to the $i$-th factor for the current reparametrization $\theta^{\phi}$.
If $\Delta$ is admissible w.r.t.\ $x_i^*$ then $\Delta$ is also admissible w.r.t.\ $x_i^{**}$.
\end{lemma}
\begin{proof}
As both $x_i^*$ and $x_i^{**}$ are optimal to $\theta^{\phi}$ and $x_i^*$ is also optimal to $\theta^{\phi + \Delta}$, we have 
$\la \Delta_{(i,j)}, A_{(j,i)} x_i^* \ra \leq \la \Delta_{(i,j)}, A_{(j,i)} x_i^{**} \ra$.
By Lemma~\ref{lemma:OptimalityAfterUpdate}, (i) also
$\la -\Delta_{(i,j)}, A_{(j,i)} x_i^* \ra \leq \la -\Delta_{(i,j)}, A_{(j,i)} x_i^{**} \ra$ holds,
hence equality must hold.
This shows $x^{**}_i \in \argmin_{x_i \in \SX_i} \la \theta^{\phi + \Delta},x_i \ra$.
Second, Lemma~\ref{lemma:OptimalityAfterUpdate}, (ii) implies that $\Delta(s) = 0$ whenever $\nu^*(s) \neq \nu^{**}(s)$.
This proves that 
$ \Delta_{(i,j)}(s) 
\begin{cases}
   \ge 0, & \nu^{**}(s) =1\\
   \le 0, & \nu^{**}(s) =0\\
\end{cases}, \nu^{**}:=A_{(i,j)} x^{**}_i$.
\end{proof}
\begin{proof}[Proof of Theorem~\ref{thm:MarginalConsistencyFixPoint}]
It is sufficient to show that for marginally consistent $\theta^{\phi}$ for $\BS$, the update $\Delta$ computed by Algorithm~\ref{alg:MessagePassingUpdate} on an arbitrary factor $i \in \BF$ and some set $J \subset \SN_{\BG}(i)$ has the following properties:
(i) $L(\phi) = L(\phi + \Delta)$,
(ii) $\theta^{\phi + \Delta}$ is marginally consistent for $\BS$.
For an easier proof, we only consider the case $J = \{j\}$. The general case can be proven analoguously.

(i)
Let $x_i^* \in \BS_i$, $x_j^* \in \BS_j$ with $A_{(i,j)} x_i^* = A_{(j,i)} x_j^*$.
We have to show that 
\begin{equation}
\label{eq:MarginalConsistentMinimaAfterUpdate}
   \min_{x_i \in \SX_i} \la \theta_i^{\phi}, x_i \ra + \min_{x_j \in \SX_j} \la \theta_j^{\phi}, x_j \ra 
   =
   \min_{x_i \in \SX_i} \la \theta_i^{\phi + \Delta}, x_i \ra + \min_{x_j \in \SX_j} \la \theta_j^{\phi + \Delta}, x_j \ra 
\end{equation}
Due to
$x^*_i$ optimal to $\theta_i^{\phi + \Delta}$, since by Lemma~\ref{lemma:UpdateSignsMultipleMinima} the update $\Delta$ is admissible for $x_i^*$,
it remains to show that 
$x_j^* \in \argmin_{x_j \in \SX_j} \la \theta^{\phi + \Delta}, x_j \ra$.
As $x_j^* \in \argmin_{x_j \in \SX_j} \la \theta^{\phi}, x_j \ra$, it is sufficient to prove that 
$x_j^* \in \argmin_{x_j \in \SX_j} \la -\Delta_{(i,j)}, A_{(j,i)} x_j \ra$.
This follows from Lemma~\ref{lemma:OptimalityAfterUpdate}~(i).
We conclude by noting
$\la \theta^{\phi}_i, x_i^* \ra + \la \theta^{\phi}_j x_j \ra = \la \theta^{\phi+\Delta}_i, x_i^* \ra + \la \theta^{\phi+\Delta}_j x_j \ra $.

(ii) 
The computations in (i) show that 
$\BS_i \subseteq \argmin_{x_i \in \SX_i} \la \theta^{\phi + \Delta}_i, x_i \ra$ and
$\BS_j \subseteq \argmin_{x_j \in \SX_j} \la \theta^{\phi + \Delta}_j, x_j \ra$.
The reparametrizations of all other factors stay the same: $\theta^{\phi + \Delta}_k = \theta^{\phi}_k$ for $k \in \BF \backslash \{i,j\}$.
Hence, $\theta^{\phi + \Delta}$ is marginally consistent for $\BS$ after the update.
\end{proof}

\section{Special Cases: Graphical Model Solvers}
\label{sec:SpecialCases}
We will show how Algorithm~\ref{alg:AbstractMessagePassing} subsumes known message-passing algorithms MSD~\cite{Werner07}, TRWS~\cite{TRWSKolmogorov}, SRMP~\cite{SRMPKolmogorov} and MPLP~\cite{MPLP} for MAP-inference with common graphical models, considered in Example~\ref{example:LocalPolytope}.

\myparagraph{Solver Primitives~\eqref{eq:FactorOptimization} and~\eqref{eq:ReverseFactorOptimization}.}
As it can be seen, all factors in~\eqref{eq:ILP:linear-constraints} are of the form 
${\SX_i = \{(1,0,\ldots,0),(0,1,0,\ldots,0),\ldots,(0,\ldots,0,1)\}}$ 
and 
$\conv(\SX_i) = \{ \mu \geq 0 : \la \eins,\mu\ra = 1\}$ is a $\dim_(X_i)$-dimensional simplex.

In all message passing algorithms~\cite{TRWSKolmogorov,SRMPKolmogorov,Werner07,MPLP}, there are two types of invokations of Algorithm~\ref{alg:MessagePassingUpdate} together with solutions of the accompanying optimization problem~\eqref{eq:FactorOptimization} and~\eqref{eq:ReverseFactorOptimization}:
{\small
\begin{tabular}{|@{\hspace{-0pt}}c|c|@{\hspace{-0pt}}c|}
\hline
\parbox[t]{0.15\linewidth}{\centering Alg.~\ref{alg:MessagePassingUpdate} \\ input} & 
\parbox[t]{0.15\linewidth}{\centering Factor\\ Optimization \eqref{eq:FactorOptimization}} & 
\parbox[t]{0.2\linewidth}{\centering Reparametrization\\ adjustment~\eqref{eq:ReverseFactorOptimization}} \\ \hline
\parbox[t]{0.15\linewidth}{\centering ${i = u \in \FV }$ \\ ${ J = \{uv\} }$ \\ ${uv \in \FE}$} &  
\parbox[t]{0.3\linewidth}{\centering ${\min\limits_{x_u \in \SX_u}\{ \theta_u^{\phi}(x_u) \}}$} & 
\parbox[t]{0.35\linewidth}{\centering{
  \begin{multline*}
  \Delta^*_{(u,uv)}(x_u) = \\
    \min_{x'_u \in X_u} \theta_u^{\phi}(x'_u)\\
    -\theta_u^{\phi}(x_u) 
  \end{multline*}
} }
  \\ \hline
\parbox[t]{0.17\linewidth}{\centering ${i = uv \in \FE}$\\ ${J = \{u\}}$\\ ${u \in \FV}$} & 
\parbox[t]{0.3\linewidth}{\centering ${\min\{\theta_u^{\phi}(x_u,x_v) \}}$ \\ ${(x_u,x_v) \in \SX_u \times \SX_v}$} & 
\parbox[t]{0.35\linewidth}{\centering{
  \begin{multline*}
    {\Delta^*_{(uv,u)}(x_u)=} \\
    {\min_{x'_{uv} \in X_{uv}} \theta_{uv}^{\phi}(x'_{uv})} \\
    -\min\limits_{x_v \in \SX_v} \{\theta_{uv}(x_u,x_v)\}
  \end{multline*}   
}}  
  \\ \hline
\end{tabular}
}

\myparagraph{MAP-inference Solvers.}
In Table~\ref{table:MapInferenceSolvers} we state solvers MSD~\cite{Werner07}, TRWS~\cite{TRWSKolmogorov}, SRMP~\cite{SRMPKolmogorov} and MPLP~\cite{MPLP} as special cases of our framework.
Factors are visited in the order they are read in.
\begin{table*}
\begin{tabular}{|ccccc|}
\hline
Algorithm & \pbox{3cm}{Current\\ factor} & $J_{receive}$ & $J_1\dcup \ldots\dcup J_{l}$ & $\omega$ \\ \hline
\multirow{2}{*}{MSD~\cite{Werner07}} 
   & $u \in \FV$ & $\SN_{\BG}(u)$ & $\{uv\} \subset \SN_{\SG}(u)$ & $\omega_1,\ldots = \nicefrac{1}{\abs{\SN_{\FG}(u)}}$\\
& $uv \in \FE$ & $\varnothing$ & --- & --- \\ \hline
\multirow{2}{*}{MPLP~\cite{MPLP}}
& $u \in \FV$ & $\varnothing$ & --- &  --- \\
   & $uv \in \FE$ & $\{u,v\}$ & $\{u\},\{v\}$ & $\omega_1 = \nicefrac{1}{2} = \omega_2$ \\ \hline
\multirow{5}{*}{\pbox{2cm}{ TRWS~\cite{TRWSKolmogorov}\\ SRMP~\cite{SRMPKolmogorov} }}
& \multicolumn{4}{c|}{forward pass:} \\
   & $u \in \FV$ & $\{uv: v \in \SN_{\SG}(u) , v < u \} $ & $ \{uv\}: v \in \SN_{\SG}(u), v > u $ & $\omega_1,\ldots = \nicefrac{1}{\max(\{v \in \SN_{\FG}(u): v>u\},\{v \in \SN_{\FG}(u): v<u\})} $\\
& \multicolumn{4}{c|}{backward pass:} \\
& $u \in \FV$ & $\{uv: v \in \SN_{\SG}(u) , v > u \} $ & $ \{uv\}: v \in \SN_{\SG}(u), v < u $ & $\omega_1,\ldots = \nicefrac{1}{\max(\{v \in \SN_{\FG}(u): v>u\},\{v \in \SN_{\FG}(u): v<u\})} $ \\
& $uv \in \FE$ & $\varnothing$ & --- & --- \\ \hline
\end{tabular}
   \caption{\cite{Werner07,TRWSKolmogorov,SRMPKolmogorov,MPLP} as special cases of Algorithm~\ref{alg:AbstractMessagePassing}.}
\label{table:MapInferenceSolvers}
\end{table*}

\begin{remark}
We have only treated the case of unary $\theta_u, u \in \FV$ and pairwise potentials $\theta_{uv}, uv \in \FE$ here. 
MPLP~\cite{MPLP} and SRMP~\cite{SRMPKolmogorov} can be applied to higher order potentials as well, which we do not treat here.
SRMP~\cite{SRMPKolmogorov} is a generalisation of TRWS~\cite{TRWSKolmogorov} to the higher-order case.
\end{remark}

\begin{remark}
There are convergent message-passing algorithms such that factors comprise trees~\cite{SubproblemTreeCalibrationWang,BlockTreeCoordinateUpdate}.
Their analysis is more difficult, hence we omit it here.
\end{remark}

Note that our framework generalizes upon~\cite{TRWSKolmogorov,SRMPKolmogorov,MPLP,Werner07,BlockTreeCoordinateUpdate,SubproblemTreeCalibrationWang} in several ways:
(i)~Our factors need not be simplices or trees.
(ii)~Our messages need not be marginalization between unary/pairwise/triplet/$\ldots$ factors.
(iii)~We can compute message updates on more than one coupling constraint simultaneously, i.e. we may choose $J_1\dcup\ldots\dcup J_l$ in Algorithm~\ref{alg:AbstractMessagePassing} to be different than singleton sets.
(i)~and~(ii) affect LP-modeling, 
(iii)~affects computational efficiency: By considering multiple messages at once in Procedure~\ref{alg:MessagePassingUpdate}, we may be able to make larger updates $\Delta^*$, resulting in faster convergence.

\end{document}